\documentclass[12pt]{amsart}

\usepackage{enumerate}
\usepackage[pdftex]{graphicx}
\usepackage{amsmath,amsthm,amssymb,latexsym,amscd}
\usepackage{setspace, bbding}
\usepackage{mathrsfs,amsfonts,wasysym,textcomp}
\usepackage{array}

\usepackage{enumerate}
\usepackage[pdftex]{graphicx}
\usepackage{amsmath,amsthm,amssymb,latexsym,amscd}
\usepackage{setspace, bbding}
\usepackage{mathrsfs,amsfonts,wasysym,textcomp}
\usepackage{array}
\usepackage[dvipsnames]{xcolor}
\usepackage{cite}

\usepackage{epstopdf}

\theoremstyle{definition}
\newtheorem{conjecture}{Conjecture}
\newtheorem{problem}{Problem}
\newtheorem{remark}{Remark}
\newtheorem{theorem}{Theorem}
\newtheorem{lemma}{Lemma}

\begin{document}

\title{Searching for Point Locations using Lines}
\author{Michelle Cordier \textsuperscript{1}, Meaghan Wheeler \textsuperscript{2}}

\thanks{\textsuperscript{1} Department of Mathematics, Chatham University, 1 Woodland Rd, Pittsburgh, PA 15232, U.S.A. E-mail: M.Doyle@chatham.edu}         

\thanks{\textsuperscript{2} University of Miami, 1320 S. Dixie Hwy, Coral Gables, FL 33146, U.S.A. E-mail: mmw134@miami.edu} 

\begin{abstract}
Versions of the following problem appear in several topics such as Gamma Knife radiosurgery, studying objects with the X-ray transform, the 3SUM problem, and the $k$-linear degeneracy testing.  Suppose there are $n$ points on a plane whose specific locations are unknown.  We are given all the lines that go through the points with a given slope.  We show that the minimum number of {\it slopes} needed, in general, to find all the point locations is $n+1$ and we provide an algorithm to do so.
\end{abstract}

\keywords{slopes, lines, point locations.  \newline {\it Subject Classification.} 52-02 Convex and Discrete Geometry, Research exposition}

\maketitle

\section{Introduction}

The following question is related to problems that appear in Gamma Knife radiosurgery, studying objects with the X-ray transform, the 3SUM problem, the $k$-linear degeneracy testing, and creating a way to search for ships in the ocean.  This subject also corresponds with geometric tomography.

\begin{problem}\label{points}
Suppose there are $n$ points on the plane, $p_1,p_2,$ $\ldots, p_n$, whose specific locations are unknown. However, we are given all the lines with a specific slope $m$ that go through all points $p_i$.  What is the minimum number of slopes needed to guarantee that the $n$ points are located?
\end{problem}

To clarify the statement, consider the following example.  Suppose that there are three points to be located and we consider the vertical and horizontal lines that pass through these points, see Figure \ref{fig:3ships_examples} for two examples.  In general there will be 3 vertical and 3 horizontal lines that intersect at 9 points.  Thus, there are 9 possible places where the 3 points could be located.  We cannot determine their exact location given only the vertical and horizontal lines that pass through the points.  However, the minimum number of slopes required is greater than 2. Next consider lines of slope 1.  Then the point locations would be found in Figure \ref{fig:3ships_examples}(left), showing that 3 slopes are required in this scenario.  However, in Figure \ref{fig:3ships_examples}(right) an additional slope would need to be considered.  In fact, if we consider the lines of slope $-1$, then the point locations would be revealed, and we conclude that four slopes are required in this case.

%Figure explaining 3 points.
    {\centering
    \begin{figure}[h!]
        \centering
        \includegraphics[width=2.25in]{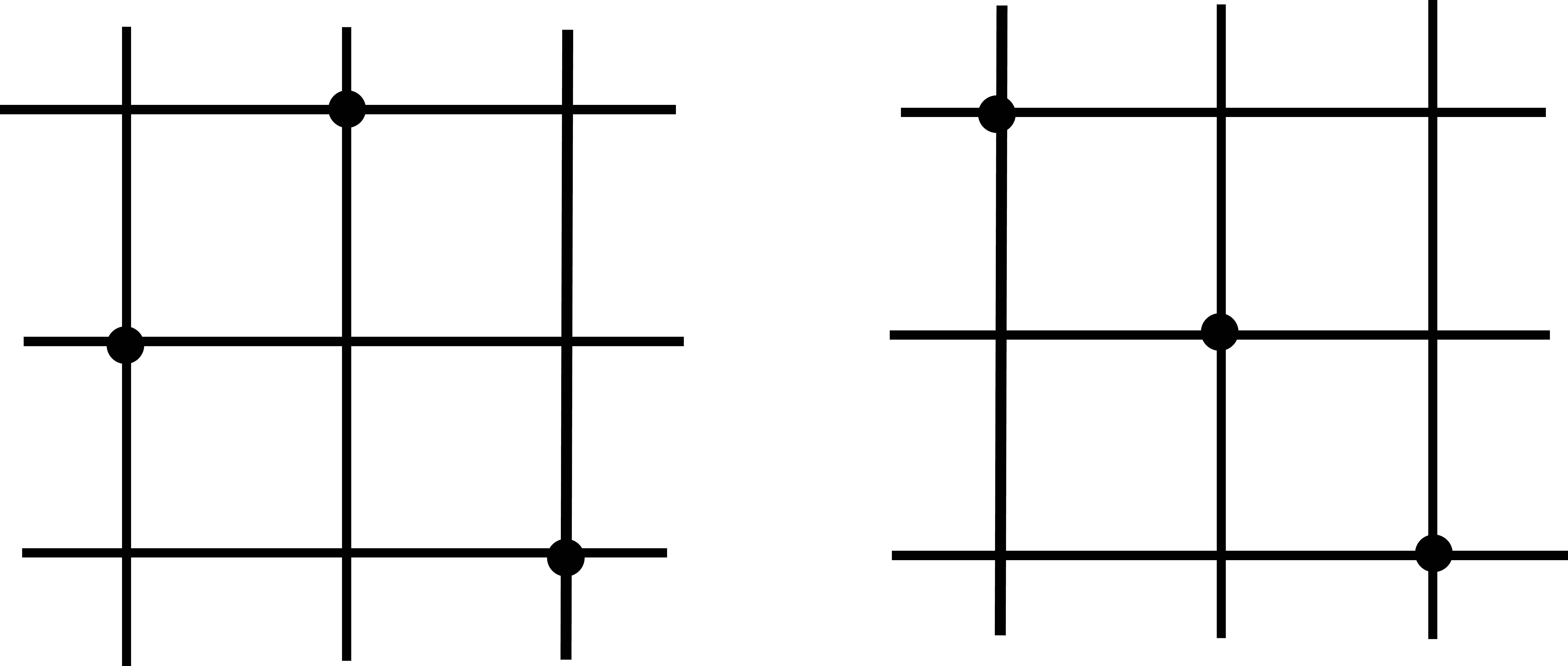}
        \caption{Horizontal and vertical lines given 3 points with potential point points of 9.}
        \label{fig:3ships_examples}
    \end{figure}}

We note that the number of slopes required is the important concept, the exact slopes are not. 

After considering several cases, it is natural to create the following conjecture.

\begin{conjecture} \label{conj}
    Suppose there are $n$ points on the plane.  Then the minimum number of slopes needed in general to guarantee the points are located is $n+1$.
\end{conjecture}

We prove that the conjecture is true for all values of $n$ (see Theorem \ref{tieredsystem} in Section \ref{main}) under the additional assumption that the slopes come from the tiered system described in the Definitions and Examples Section (Section \ref{definitions}).  In addition, we prove that if there are $n$ points, then the minimum number of slopes needed in general, without using the tiered system, is less than or equal to $n+1$ (see Theorem \ref{general} in Section \ref{main}).

In Section \ref{applications} we discuss the applications. Next, in Section \ref{definitions} we introduce notation and definitions, consider specific examples, and create a table showing the amount of different configurations.  Section \ref{intro_thm_lem} shows that Conjecture 1 is true for the values of $n=1,2,3$ with the added stipulation that the vertical and horizontal lines are separated by 1 unit length.  Then in Section \ref{main}, we prove that the minimum number of slopes needed in general is less than or equal to $n+1$ (Theorem \ref{general}), and if adding the stipulation of the tiered system, then Conjecture 1 is proved for any $n$ (Theorem \ref{tieredsystem}).

\section{Applications} \label{applications}
%Applications
As stated previously, there are several topics that have provided motivation for this problem.  In particular, similar problems occur in Gamma Knife radiosurgery, studying objects with the X-ray transform, the 3SUM problem, more specifically the $k$-linear degeneracy testing, and creating a way to search for ships on the ocean.

Gamma Knife radiosurgery is a type of procedure that is widely used by medical professionals to treat brain abnormalities such as brain tumors or malformations \cite{Mayo}. Gamma Knife radiosurgery is unique in the sense that it does not require any incisions. Instead, Gamma Knife uses pinpoint precision lasers to provide targeted radiation to the abnormality in the brain. This method has been proven effective and safer because the radiation is only applied to the abnormality reducing the amount of tissue harmed by radiation. Our algorithm could potentially be applied to this medical application of targeting brain abnormalities, due to the fact of directing the lasers in specific directions. 

Another problem occurs when directing a X-ray beam.  The {\it directed X-ray at a point} \cite[p. 195]{Gardner} is defined to be the length of all the rays within a body $K$ when measured from a certain point.  Thus, instead of considering lines that go through a set of points as in our case, it indicates the length of a body in all directions from a specific point.

A similar question which is related to Problem 1 is the 3SUM problem \cite{Gronlund}.  The 3SUM problem asks whether given a set of real numbers, are there three elements of that set that add to 0?  Stated more precisely, given a set $A \subset \mathbb{R}$, determine if there exists $a,b,c \in A$ such that $a+b+c=0$.  Finding $a,b,c$ is akin to finding equations of lines that all go through a point.

Additionally, the $k$-linear degeneracy testing ($k$-LDT), \cite{Gronlund}, supposes we are given a $k$-variate linear function $\phi$, and a set of real numbers $A$, and asks if there exists ${\bf x} \in A^k$ such that $\phi({\bf x})=0$.  Here, $A^k$ means the Cartesian product of $A$, $k$ times.  A detailed statement of this problem is: Fix a $k$-variate linear function $\phi(x_1,\dots,x_k)=\alpha_0+\sum_{i=1}^k \alpha_i x_i$, where $\alpha_0,\ldots,\alpha_k \in \mathbb{R}$, then, given a set $A\subset \mathbb{R}$, determine if $\phi({\bf x})=0$ for any ${\bf x} \in A^k$.  This is similar to our problem of having linear functions that all go through specific points, then determining at which points the lines are concurrent.

Finally, we describe a way to locate ships in the ocean.  This is a specific example that could be applied in various other scenarios.  Suppose there are missing ships in the ocean and we need to locate them by using a spotlight.  The spotlight can only travel in a straight line and cannot tell us the distance between it and the object.  We then suppose there is a moving border ship that points the spotlight in a specific direction (prescribing the slope of the line) and records its own location at the moment when the spotlight detects a missing ship.  The process is repeated in the other directions.  Our algorithm guarantees that the border ship would be able to find the missing ships with at most a specific number of directions for the spotlight.

\section{Definitions and Examples} \label{definitions}

%Definitions
The number of points will be denoted by $n$, and the minimum number of slopes needed in general for a given $n$ will be denoted by $s$.  We define a {\it potential point} to be a location where there may or may not be a point in our configuration.

Due to the importance in finding the minimum number of slopes required and not the specific slopes, we create a tiered system.\\

{\bf Tier 1}: $\{$horizontal, vertical$\}$ 

{\bf Tier 2}: $\{ 1 , -1 \}$ 

{\bf Tier 3}: $\{ \pm 1/2, \pm 2 \}$

{\bf Tier 4}: $\{ \pm 1/4, \pm 4, \pm 3/4, \pm 4/3 \}$

$\vdots$ 

{\bf Tier n}: $\{ \pm k/2^{n-2}, \pm 2^{n-2}/k : k \mbox{ is odd and } 0 < k < 2^{n-2} \}$ 

$\vdots$ \\

The tiered system works as follows.  If two slopes are required, then only the slopes from Tier 1 will be considered.  If three slopes are required, then two of the slopes will be from Tier 1 and the last slope will be from Tier 2 (either 1 or $-1$).  If 4 slopes are required, then all slopes from Tier 1 and Tier 2 will be considered.

We first note that Tier 1 is chosen without loss of generality, any two slopes can be linearly transformed into vertical and horizontal lines.  Tier 2 was chosen to have the angles and lines as separate as possible from the lines in Tier 1.  This is to address the real world example dealing with tumors or ships, we would not know how wide a ``point'' (a tumor, or ship) would be and thus would want the biggest angle between slopes to guarantee that we are not looking at the same ``point''.  With similar reasoning, we choose the remaining tiers.

%Examples
Let us consider various examples related to Problem \ref{points}.  First we look at Figure \ref{fig:3ships_examples} in more detail.  Three points are given, and when the vertical and horizontal lines through the points are considered, there are 9 potential points.  If the lines with a slope from Tier 2 are added, (for example the slope 1), then in the case of Figure \ref{fig:3ships_examples} (left) the points would be located.  On the other hand, in Figure \ref{fig:3ships_examples} (right) there are still 5 potential points and an additional slope is required.  Thus, consider lines of slope $-1$, this would reveal the locations of the points, concluding that 4 slopes are required. 

%Put another figure with slope 1 lines.  Then reference it.

Sometimes it is the case that choosing the slopes in a different order would  require a smaller amount of slopes to reveal the location of the points.  This is the case with Figure \ref{fig:3ships_examples} (right).  The points would be located using 3 slopes (instead of 4) if the vertical and horizontal lines are considered and then the lines with slope $-1$.

In Figure \ref{fig:3pointsexample2}, more 3 point configurations are considered.  After drawing the vertical and horizontal lines, in Figure \ref{fig:3pointsexample2} (left) and (middle) there are 6 potential points, and in Figure \ref{fig:3pointsexample2} (right) there are 3 potential points.  This implies that in Figure \ref{fig:3pointsexample2} (right) the 3 points are revealed and only 2 slopes are required in this example.

For Figure \ref{fig:3pointsexample2} (left), now consider lines with slope 1, then there would be 4 potential points and hence an additional slope is required.  Considering lines with slope $-1$ would then reveal the exact point locations.  A similar argument works for Figure \ref{fig:3pointsexample2} (middle).  This indicates that with the configurations analyzed so far, if $n=3$ then $s \geq 4$, in general.

%Figure explaining 3 points with a smaller amount of intersections.
    {\centering
    \begin{figure}[h!]
        \centering
        \includegraphics[width=3in]{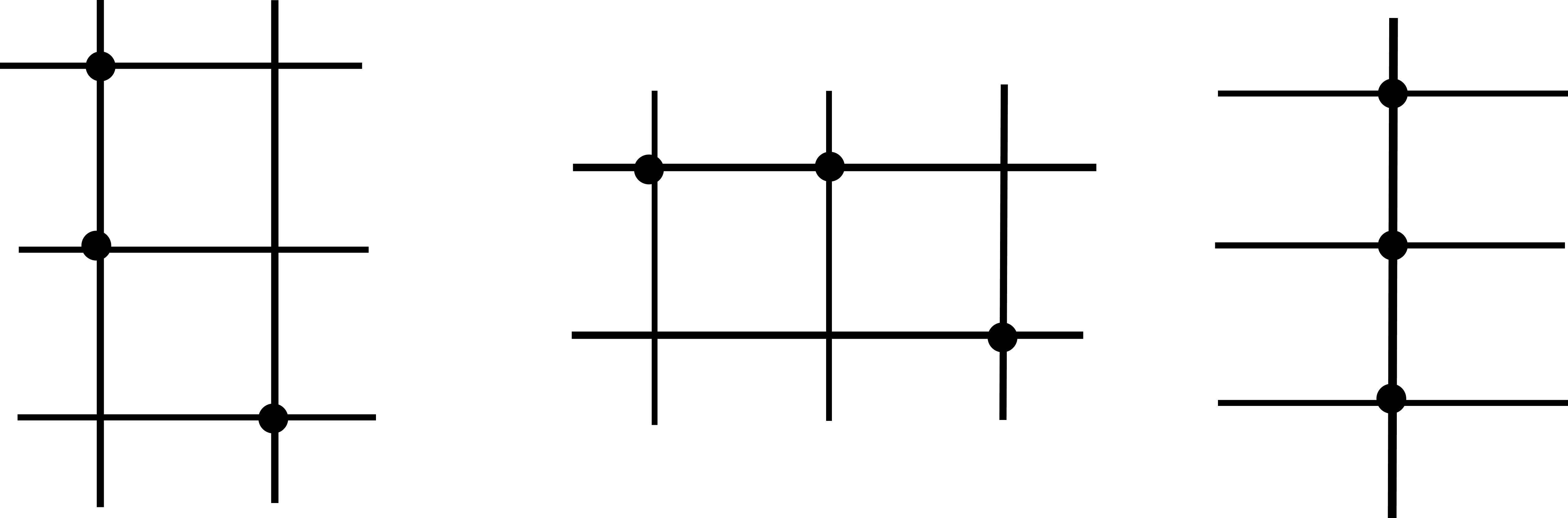}
        \caption{Given 3 points with the horizontal and vertical lines drawn.  The number of potential points are 6,6, and 3 respectively.}
        \label{fig:3pointsexample2}
    \end{figure}}

Table \ref{table1} illustrates the number of possible examples containing up to 4 points.  For this table the following slope order is used: horizontal, vertical, 1, $-1$, 1/2, $-2$, 2, $-1/2$.  This slope order is used to create the table so that there is a methodical way of computing its entries.  In Section \ref{intro_thm_lem}, the proofs of the theorems use our tiered system and do not abide by this order.  In the last column of the table the number of examples that are associated with each specific number of slopes required from the previous column is computed.

\begin{table}
$$\begin{array}{|c|c|c|c|}
    \hline \mbox{Points} & \mbox{Examples} & \mbox{Slopes required} & \mbox{Number of examples} \\
    & & & \mbox{requiring those slopes} \\
    \hline
    1 & 1 & 2 & 1 \\ \hline
    2 & 4 & 2 & 2 \\
      &   & 3 & 2 \\ \hline
    3 & 24 & 2 & 2 \\ 
    & & 3 & 9 \\
    & & 4 & 13 \\ \hline
    4 & 196 & 2 & 3 \\
    & & 3 & 23 \\
    & & 4 & 104 \\
    & & 5 & 66 \\ \hline
\end{array}
$$
\caption{Number of examples given a specific amount of points.  In addition, the number of examples with a specific number of slopes required is computed.}\label{table1} 
\end{table}

After analyzing these examples for the values $n=1,2,3$, the number of configurations is small, allowing the examples to be computed by hand and to create initial strategies that will help with bigger values of $n$.  For $n=4$ the number of examples is larger, and for $n=5$, the number of examples increases to an amount not reasonable to compute by hand.  We utilize a different approach to prove our main results Theorem \ref{general} and \ref{tieredsystem}, which will be discussed in Section \ref{main}.

\section{Introductory Theorems} \label{intro_thm_lem}

In this section we utilize the tiered system and assume that the horizontal and vertical lines are separated by a unit distance.  We prove Conjecture \ref{conj} for $n=1,2,3$ in this special scenario using a case by case argument. 

To begin, assume there is one point on the plane.

\begin{theorem} \label{one} If there is one point, then two slopes are required to locate that point.
\end{theorem}

\begin{proof} 
Suppose there is one point.  Without loss of generality, consider the horizontal slope from Tier 1.  There will be one horizontal line that goes through the point.  The position of the point from this one slope cannot be determined, thus the second slope from Tier 1, namely the vertical slope, must be considered.  There will be one vertical line that goes through the point and which intersects the horizontal line.  This point of intersection is the point we are looking for. 

%{Denote the unknown point as $(x_1,y_1)$.  Picking slopes from Tier 1, consider the horizontal line $y=y_1$. The point is located somewhere on this line, but we do not know exactly where.  Thus we must consider the vertical line $x=x_1$.  We can now determine where the point $(x_1,y_1)$ is located from the intersection of the two lines.}
\end{proof}

        {\centering
    \begin{figure}[h!]
        \centering
        \includegraphics[width=1.75in]{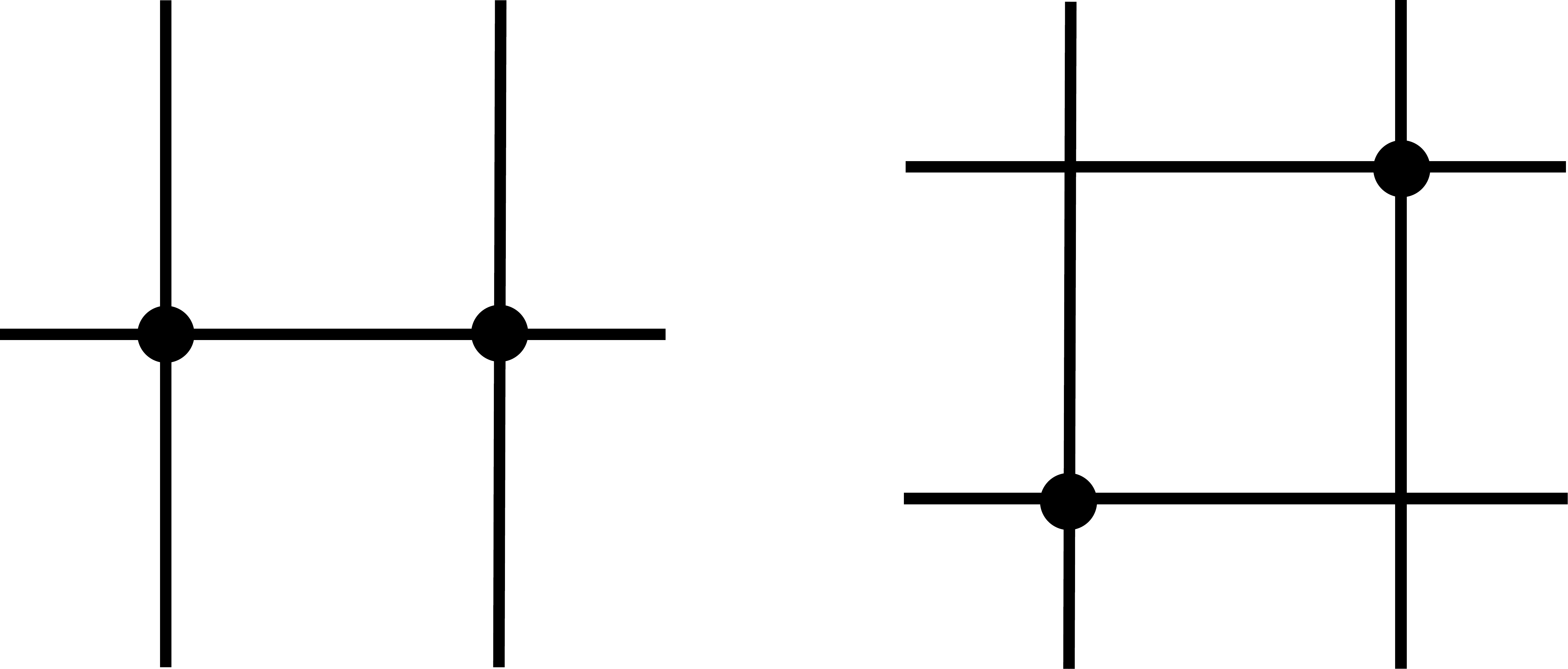}
        \caption{The required number of slopes needed is two on the left and three on the right.}
        \label{fig1}
    \end{figure}}
    
Next we consider if there are two points.

\begin{theorem} \label{two}
If there are two points, then the minimum number of slopes needed in general is three.
\end{theorem}

\begin{proof}

Suppose there are two points on the plane and consider the horizontal and vertical lines through those points.  There may be 2 or 4 potential points.  If there are 2 potential points, see Figure \ref{fig1} (left) for an example, then those points are in fact the actual points in the configuration.  The remaining case is when there are 4 potential points.  Consider the example in Figure \ref{fig1} (right).  Then an additional slope is required and either slope from Tier 2 will suffice.  The only additional configuration in this case can be obtained by rotating the configuration in Figure \ref{fig1} (right) by 90 degrees, and the same conclusion is reached.

\end{proof}

\begin{remark}
    The tiered system assists in creating an algorithm to find the points in Theorem \ref{one} and \ref{two}.  However, generic slopes could be used to complete these proofs.
\end{remark}

Next, there is an important Lemma that will assist us in proving Theorem \ref{three}.  It shows that the number of slopes required to reveal a configuration is invariant under rotations and reflections.  Notice, this Lemma is true in general and does not have any added stipulation on the distance between horizontal and vertical lines.

%Other lemmas were placed here.  Now I put them after the end of the document.

\begin{lemma}\label{rr}
If there is a configuration that reveals the locations of the points with $t$ number of slopes, then the following transformations of the configuration will reveal its points with $t$ number of slopes as well:
\begin{enumerate}[(i)]
    \item rotations by 90, 180, 270 degrees counterclockwise,
    \item horizontal or vertical reflections through the center of the configuration.
\end{enumerate}
\end{lemma}

\begin{proof}
(i) First consider rotating the configuration 90 degrees counterclockwise.  Now the lines that have horizontal and vertical slopes will become vertical and horizontal, respectively, thus staying in Tier 1.  The lines with slopes 1 and $-1$ will become lines with slopes $-1$ and 1, respectively.  Finally, lines with slope $k/2^{n-2}$, $2^{n-2}/k$, $-k/2^{n-2}$, $-2^{n-2}/k$ will become lines with slope $-2^{n-2}/k$, $-k/2^{n-2}$, $2^{n-2}/k$, $k/2^{n-2}$.  Hence, after rotation the lines within the specific tier will remain in that tier.  If a configuration reveals the points with $t$ number of slopes, the rotated configuration will reveal the points with the lines also being rotated, and will require $t$ number of slopes.

Rotating the configuration and the lines by 180 degrees will also reveal the points.  We note that all lines with slope $m$ will keep the same slope when rotated.  Thus in this case, the same lines will reveal the rotated configuration.

Now consider rotating the configuration 270 degrees counterclockwise.  The rotated lines will become the exact same lines when considering the rotation by 90 degrees counterclockwise.  Thus the configuration is revealed.

(ii) Consider the horizontal reflection through the center of the configuration.  The horizontal lines will remain horizontal, and the vertical lines will remain vertical.  Lines with slope 1 and $-1$ will reflect to be lines with slope $-1$ and 1, respectively.  Similarly, lines with slope $k/2^{n-2}$, $2^{n-2}/k$, $-k/2^{n-2}$, $-2^{n-2}/k$ will become $-k/2^{n-2}$, $-2^{n-2}/k$, $k/2^{n-2}$, $2^{n-2}/k$, respectively.  Again, the slopes remain in their specific tier.  Thus, if a configuration is revealed using $t$ slopes, the reflected configuration will be revealed using the reflected lines and using $t$ slopes.

The argument for the vertical reflection is similar.
\end{proof}

We note that when a configuration is rotated or reflected, it may require a smaller amount of slopes to reveal the point's locations.  The important thing to realize is that Lemma \ref{rr} states that $t$ is at least sufficient in revealing the configuration.

%Corollary that used phantom points and rotations and reflections.  We do not reference it, so it is now after the end of the document.

%Other lemma and theorem were here.  But now are after the end of the document.

Now we are ready to consider the case of 3 points on the plane.  We break the proof into 4 cases, discuss various examples, and utilize the geometry of the cases to find the number of slopes required in all situations.

\begin{theorem} \label{three}
    If there are three points, then the minimum number of slopes needed in general is four.
\end{theorem}

\begin{proof}
Consider a configuration of $n=3$ points, and the horizontal and vertical lines for that example.  There are four different cases for the number of potential points, namely 3,4,6, or 9.  In the figures for this proof, we typically exclude reflections or rotations of configurations due to Lemma \ref{rr}.

{\bf Case 1: 3 potential points.} The location of the points are revealed with only horizontal and vertical lines.  See for example Figure \ref{fig:3points4Intersections} (left).  In this case the number of slopes required is two.

{\bf Case 2: 4 potential points.} 
This implies there will be two horizontal lines and two vertical lines, each pair being separated by 1 unit length by assumption.  An additional slope is required.  Consider a slope from Tier 2, say $m=1$, and the lines with that slope.  If there are only 3 potential points, then the configuration is revealed and the number of slopes required is 3.  However, if there are 4 potential points then there are 3 lines of slope 1, see Figure \ref{fig:3points4Intersections} (right) for an example.  The outer most diagonals will contain a point by default and the middle diagonal will contain exactly one point.  Thus, we know where two points are located.  When the remaining slope from Tier 2 is considered ($m=-1$ in this case), the two points we know will be connected with a single line.  The last point will be revealed from the other line with this final slope.

    {\centering
    \begin{figure}[h!]
        \centering
        \includegraphics[width=2in]{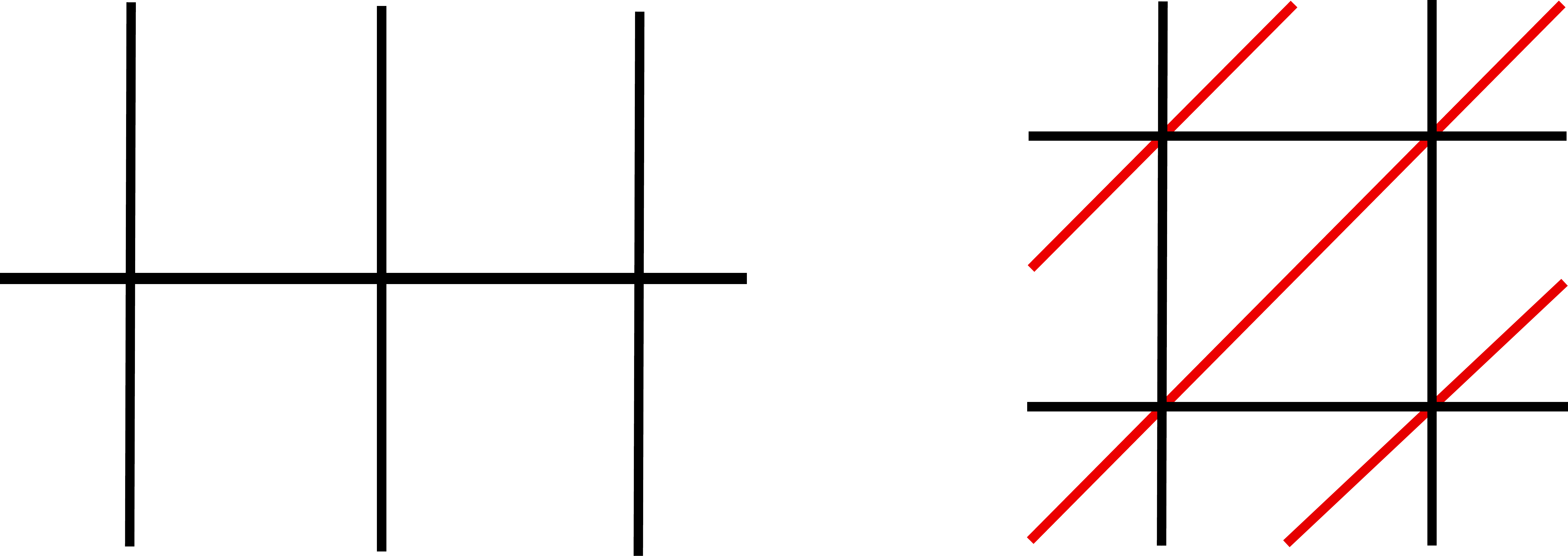}
        \caption{(left) On the left, there are 3 potential points.  On the right, there are 4 potential points with slope of 1.}
        \label{fig:3points4Intersections}
    \end{figure}}

{\bf Case 3: 6 potential points.} First consider the case where there are two rows and three columns.  The number of potential points after we consider the lines of slope 1 is either 3 or 4.  
It is impossible to have 5 or 6 potential points in this instance as we get contradictions in all cases.  Indeed, in Figure \ref{3points2Rows3Columns} (left), the only way to get 6 potential points is to have all 4 lines of slope 1 exist, but this implies that there is at least 4 points which is a contradiction.  For 5 potential points, considering Figure \ref{3points2Rows3Columns} (left) again, two inner most diagonals must be included and one outer diagonal.  Note, there must be a point in each column because there are 3 points in the configuration.  Since there is only 1 point on the outer diagonal, that must be an actual point.  The next inner diagonal reveals a point in the same row as the first point because of the one point per column stipulation, and similarly for the final inner diagonal.  This implies that all points are in the same row and the other row does not have a point on it, which is a contradiction.  A similar argument can be made if the potential points formed 3 rows and 2 columns.

If there are 3 potential points then the points are revealed.  Thus we assume there are 4 potential points after drawing the lines of slope 1.  There are two ways to have 4 potential points, namely 3 or 2 lines of slope 1, see Figure \ref{3points2Rows3Columns} (left).  In the first case of 3 lines of slope 1, the three lines must include the outer most diagonals and one inner diagonal (if both inner diagonals are used, then there is a contradiction), see Figure \ref{3points2Rows3Columns} (middle) for an example.  Thus, the two corner points, the top left and bottom right, are actual points and the point that is in the middle column is revealed from the remaining line of slope 1 (otherwise, the middle column would not have a point located on it and would not exist).  Thus, 3 slopes are required.  In the remaining case for 4 potential points, there are 2 lines of slope 1, these 2 lines must be the two inner diagonals, see Figure \ref{3points2Rows3Columns} (right) for instance.  In order for the configuration to include the left and right vertical lines, the points in the first and last column must be actual points. Furthermore, the point in the middle column cannot be determined, as there are two options, thus an additional slope is required.  Considering the lines of slope $-1$ will reveal the point in the middle column.  Hence, 4 slopes are required.

Next consider the configurations where there are three rows and two columns.  Then by Lemma \ref{rr}(i), the maximum number of slopes needed is 4. 

{\centering
    \begin{figure}[h!]
        \centering
        \includegraphics[width=3.5in]{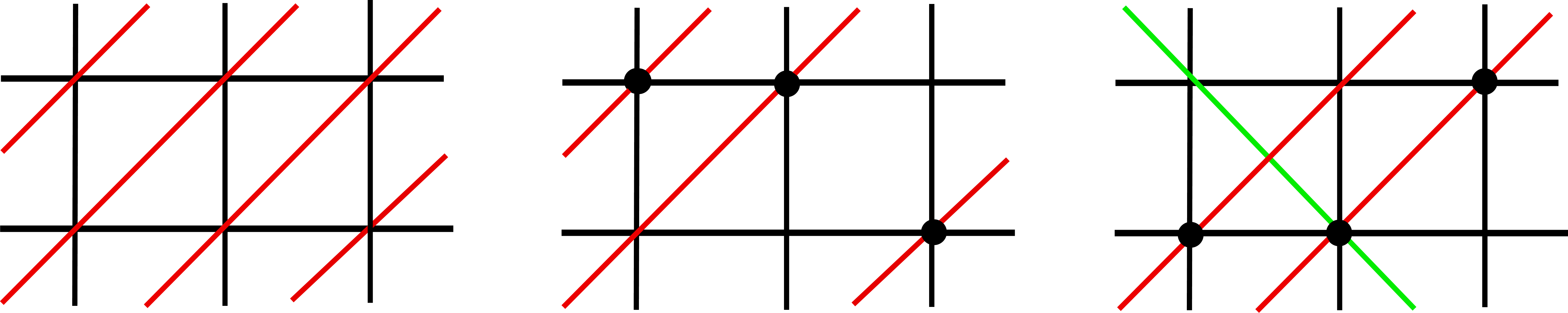}
        \caption{Potential lines of slope 1.  To have 4 potential points, there may either be 3 or 2 lines of slope 1.}
        \label{3points2Rows3Columns}
    \end{figure}}

{\bf Case 4: 9 potential points.} We note that having 9 potential points implies that each point must be in its own row and column. There are 6 possible configurations, of which only 2 are unique up to rotation and reflection, see Figure \ref{3points3Rows3Columns}. In all examples, 4 slopes is enough to reveal the locations of the points.

{\centering
    \begin{figure}[h!]
        \centering
        \includegraphics[width=2.5in]{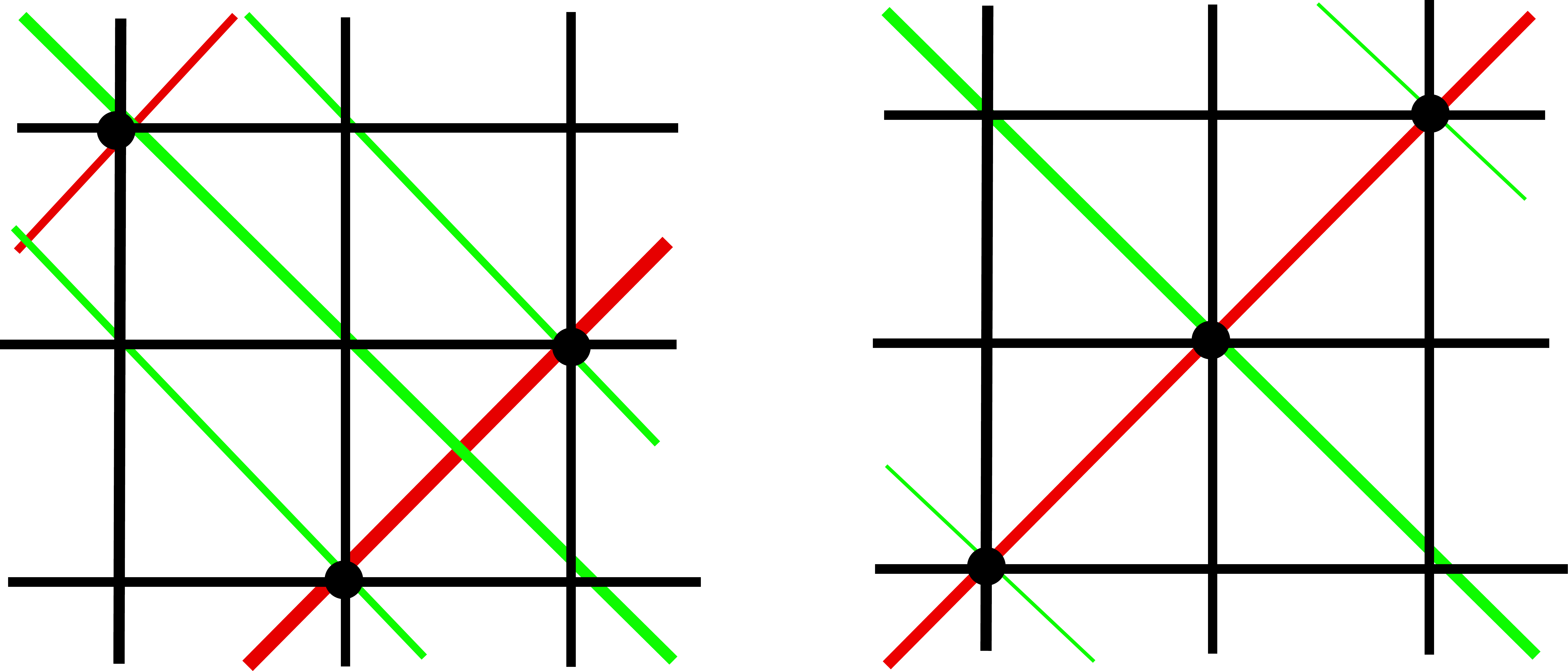}
        \caption{Configurations where $n=3$ and 9 potential points. Each point is in its own row and column.}
        \label{3points3Rows3Columns}
    \end{figure}}

\end{proof}

In the next section, we will discuss a way to avoid considering all the different cases.

\section{General Results}\label{main}

Instead of computing directly the number of slopes required for all configurations for $n \geq 4$, we work backwards.  The worse case scenario would be that we have used the claimed maximum number of slopes and there is still a potential point that is not an actual point.  We will refer to a point in this case as a {\it phantom point}.  Then we prove that there are no phantom points in any configuration by reaching a contradiction in all cases in question.  

Let us consider an example for further clarification. Suppose $n=4$ and there are 3 horizontal and vertical lines, see Figure \ref{Phantom_Point}.  We claim that $s=5$ and consider the horizontal and vertical slopes, along with slopes of 1, $-1$ and 1/2.  Finally, let us assume that there is a phantom point in the second row, first column (as pictured in Figure \ref{Phantom_Point}).  To be a phantom point all lines with the listed slopes must go through it.  Thus, considering a slope of 1, this forces a point to be located in the first row, second column.  Similarly, for a slope of $-1$, a point is forced to be in the third row, second column.  Finally, for a slope of 1/2, there must be a point in the first row, third column.  Therefore, three of the four points are forced in this configuration.  To find the last point, either the phantom point is an actual point which is a contradiction, or two points are needed, one in the first column and one in the second row which is also a contradiction since there are only 4 points.

{\centering
    \begin{figure}[h!]
        \centering
        \includegraphics[width=1.25in]{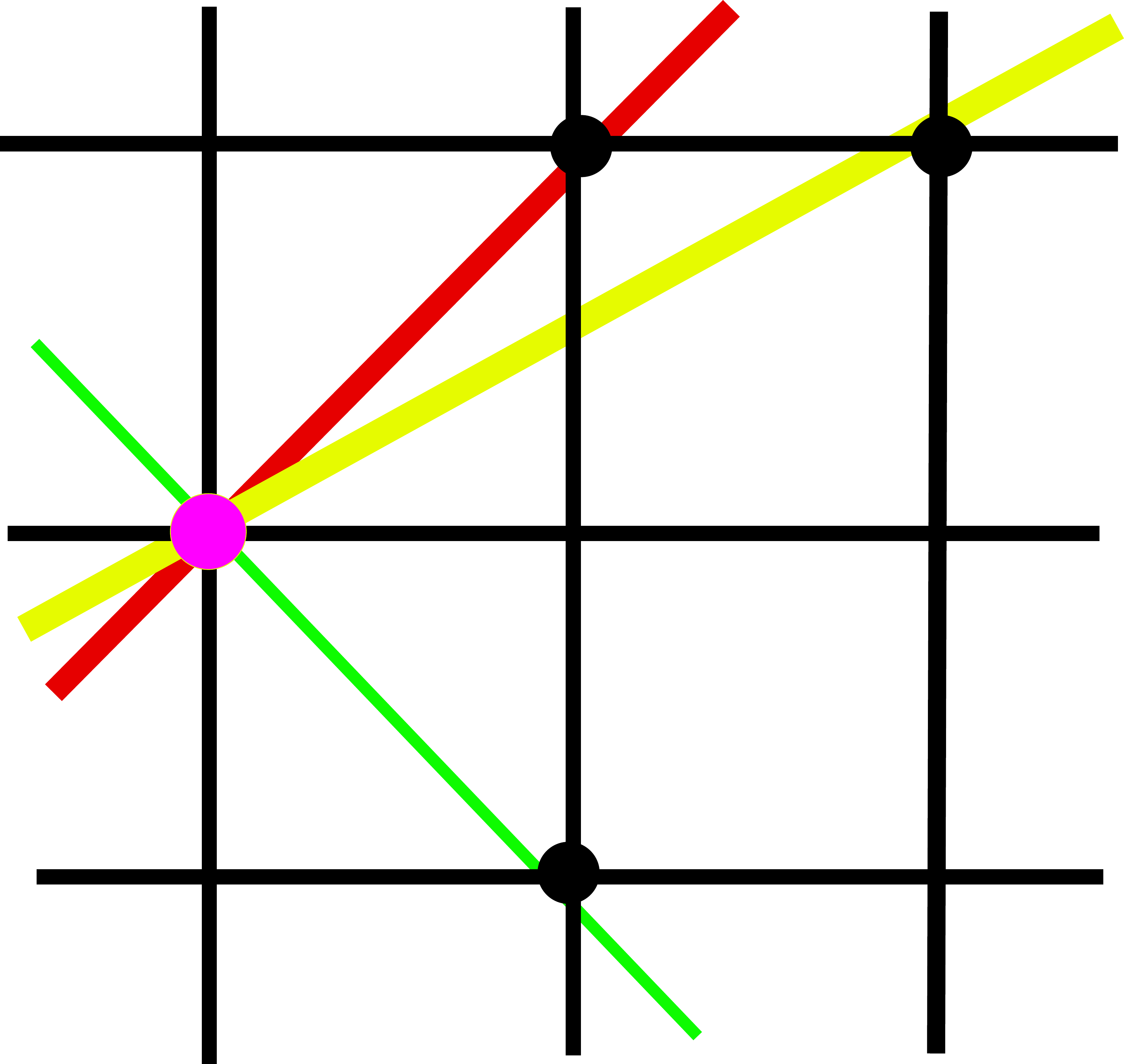}
        \caption{The pink circle is a phantom point and forces the location of 3 points.}
        \label{Phantom_Point}
    \end{figure}}

We utilize our definition of phantom point for the following proof.  Observe that we do not impose any conditions on the spacing of the points, more specifically the points are no longer expected to be separated horizontally or vertically by a unit length.  The next theorem shows that $n+1$ slopes are enough to reveal any configuration with $n$ points.  However, it does not provide an algorithm to do so.

\begin{theorem}\label{general}
If there are $n$ points, then the minimum number of slopes needed in general is $s \leq n+1$.
\end{theorem}    

\begin{proof}
Suppose there are $n$ points and $s>n+1$.  This implies that $n+1$ slopes is not enough to reveal the configuration, thus there is a phantom point when consider $n+1$ slopes.  In particular, there is a potential point that has $n+1$ slopes that go through it and is indeed not a point in the configuration.  Now there must be an actual point on each line that intersects the phantom point.  Hence, there would be $n+1$ total points to do this, which is impossible since there are only $n$ points.  Therefore, there are no phantom points when $n+1$ slopes are considered, and the minimum number of slopes needed in general is $s \leq n+1$.

\end{proof}

Next, we will use the tiered system to show that there are configurations that require $n+1$ slopes when considering $n$ points.  This will show that our result is sharp.

\begin{theorem} \label{tieredsystem}
Consider slopes from the tiered system, and assume there are $n$ points. Then the minimum number of slopes needed in general is $s=n+1$.
\end{theorem}

\begin{proof}
Consider the first $n$ slopes from the tiered system, denoted by $s_1,s_2,\ldots,s_n$.  Define the center of the configuration to be the point $z$, and suppose $z$ is not a point in the configuration.  Then place the first point in the configuration a unit distance away from $z$ on the line with slope $s_1$ that goes through $z$.  Next place the second point a unit distance away from $z$ along the line going through $z$ with slope $s_2$.  Continue this pattern until you placed all $n$ points.  Then the center of the configuration intersects all lines with the slopes $s_1,s_2,\ldots,s_n$, and is not in fact a point (all $n$ points have already been placed).  This shows that $z$ is a phantom point.  Hence, at least one more slope is needed to eliminate this point.  In particular, the number of slopes required in this example is at least $n+1$. From Theorem \ref{general}, we see that this implies that the minimum number of slopes required in general is $n+1$.
\end{proof}

We have successfully answered Conjecture \ref{conj} for all $n$ using an algorithm with a tiered system of slopes.  It is natural to consider the higher dimensional version of this problem.  How many slopes would be required according to the dimension?  In particular, the case when we are in 3-dimensions?

%%%%%%%%%%REFERENCES%%%%%%%%%%

\end{document}